\documentclass[12pt]{amsart}%
\usepackage{amsmath}
\usepackage{amssymb}
\usepackage{amsthm}
\usepackage{amscd}
\usepackage{comment}
\usepackage[dvips]{graphicx}
\usepackage[mathscr]{eucal}
\newtheorem{Thm}{Theorem}[section]
\theoremstyle{definition}
\newtheorem{Theorem}[Thm]{Theorem}
\newtheorem{Lemma}[Thm]{Lemma}

\newtheorem{Proposition}[Thm]{Proposition}

\theoremstyle{remark}
\newtheorem{Remark}{Remark}

\font\ym=msbm10  
\newcommand{\Aut}{{\rm Aut}}

\newcommand{\R}{\text{\ym R}}

\newcommand{\C}{\text{\ym C}}

\newcommand{\sB}{\mathscr B}
\newcommand{\sC}{\mathscr C}

\newcommand{\sF}{\mathscr F}
\newcommand{\sH}{\mathscr H}

\newcommand{\sP}{\mathscr P}

\title[]
{Scaling Flow on Covariance Forms of CCR Algebras}
\author[Scaling Flow]{Shigeru Yamagami}
\address{Graduate School of Mathematics\\
  Nagoya University}

\begin{document}
\maketitle   

\begin{abstract}
  In connection with parametric rescaling of free dynamics of CCR,
  we introduce a flow on the set of covariance forms 
  and investigate its thermodynamic behavior at low temperature
  with the conclusion that every free state approaches to a selected Fock state as a limit. 
\end{abstract}

\section*{Introduction}
In canonical quantum algebras, 
there is a close relationship 
between free states and one-parameter groups of Bogoliubov automorphisms (\cite{BR2}).
Rescaling of the group parameter is therefore expected to induce a flow on free states.
We shall here present an explicit realization of this idea in terms of covariance forms (or two-point functions)
parametrizing free states.
Since covariance forms constitute a convex set, it gives a flow of geometric nature.
The fixed points are characterized as extremal points and the asymptotics of the flow is then intrepreted as
expressing thermodynamic behavior at high or low temperature.

In the case of fermion algebras, covariance forms are simply expressed by covariance operators
and, given a covariance operator $C$ in a *-Hilbert space $V^\C$, the thermodynamic Hamiltonian is
the infinitesimal generator of one-parameter unitary group ${\overline C}^{it} C^{-it}$ (\cite{Ar}),
whence the corresponding $r$-th power scaling is nothing but 
$C^r/(C^r + {\overline C}^r)$ with its asymptotic behavior read off from functional limits
\[
  \lim_{r\to 0} \frac{c^r}{c^r + (1-c)^r} =
  \begin{cases} 0 &\text{if $c = 0$,}\\
    1/2 &\text{if $0 < c < 1$,}\\
    1 &\text{if $c = 1$}
  \end{cases}
\]
and 
\[
  \lim_{r\to \infty} \frac{c^r}{c^r + (1-c)^r} =
  \begin{cases} 0 &\text{if $0 \leq c < 1/2$,}\\
    1/2 &\text{if $c = 1/2$,}\\
    1 &\text{if $1/2 < c \leq 1$.}
  \end{cases}
\]

In the case of boson algebras of infinite freedom, some subtlety arises however,
because there is no privileged topology in
an infinite-dimensional symplectic vector space $(V,\sigma)$
and we need to specify an inner product and take a completion of $(V,\sigma)$ to describe one-parameter group of
one-particle transformations (\cite{AY}).

For covariance forms, though there is no overall inner product, we can introduce a flow as a functional calculus
on positive forms so that it really gives a parametric rescaling of free dynamics under a mild condition
of non-geneneracy on covariance forms.
In the case of finite-dimensional symplectic vector spaces, the relevant physics is quantum mechanics of
finite degree of freedom, which enables us to work with density operators and the scaling flow is
explicitly related to powers of density operators.

In this situation, we show that high temperature limit is divergent,  
whereas zero temperature limit freezes
each free state to a selected Fock state. 

The author is indebted to Taku Matsui for condensation related meanings of zero temperature limits
and to Masaru Nagisa for informing us about the argument method
to check operator-monotonicity of functions which appeared in describing the scaling flow.


\section{Positive Sesquilinear Forms}
A real world can be understood from a complex world with complex ones more flexible than real ones.
In linear algebra, real vector spaces are in a one-to-one correspondence with complex vector spaces furnished with
complex conjugation.
Recall that a complex conjugation in a complex vector space $K$ is a conjugate-linear involution,
which is denoted by $\overline{x}$ or $x^*$ in this paper.
Based on the second notation a complex vector space with complex conjugation
is also referred to as a *-vector space. A complex conjugation specifies a real structure
in the sense that it selects a real subspace $V = \text{Re} K = \{ x \in K; \overline{x} = x\}$
satisfying $K = V + iV$.
Conversely, given a real vector space $V$, its complexification
$V^\C = V + iV$ admits a canonical conjugation $\overline{v + iw} = v - iw$ ($v,w \in V$) so that
$V = \text{Re}(V^\C)$.
Thus there is a one-to-one correspondence between real vector spaces and *-vector spaces.

Complex conjugation can be then applied to various linear algebraic objects
such as linear maps and sesquilinear forms:
\begin{enumerate}
\item
  For a $\C$-linear map $\phi: V^\C \to W^\C$,
its complex conjugate $\overline{\phi}$ is defined by $\overline{\phi}(x) = \overline{\phi(\overline{x})}$
($x \in V^\C$), which is again a $\C$-linear map of $V^\C$ into $W^\C$.
\item
  For a sesquilinear form $S$ on $V^\C$, its complex conjugate $\overline{S}$ is defined by
  $\overline{S}(x,y) = \overline{S(\overline{x},\overline{y})}$ ($x,y \in V^\C$),
  which is again a sesquilinear form on $V^\C$. 
\end{enumerate}

A linear algebraic object is said to be \textbf{real}
if it is unchanged under the operation of complex conjugate.
Thus real $\C$-linear maps $V^\C \to W^\C$ are in a one-to-one correspondence with
$\R$-linear maps $V \to W$ and
real sesquilinear forms on $V^\C$ are in a one-to-one correspondence with $\R$-valued bilinear forms on $V$.
Note that, as for sesquilinear forms, symmetric/alternating forms on $V$ correspond to 
real forms on $V^\C$ which are hermitian/anti-hermitian. 

Let $\sP(V)$ be the set of positive sesquilinear forms on $V^\C$, which is a convex cone. 
For $S \in \sP(V)$, its complex conjugate $\overline{S}$ is again a positive form and 
$S + \overline{S}$ is a real positive form,
while $(S - \overline{S})/i$ is a real but antihermitian form. 
These correspond to a symmetric or alternating bilinear form on $V$ respectively and
called the real or imaginary part of $S$, though the correct usage should be applied to their halves.

Let $V_S^\C$ be the Hilbert space associated to the positive form $S + \overline{S}$
with the inner product of $V_S^\C$ denoted by $(\ ,\ )_S$.
The complex conjugation in $V^\C$ gives rise to an anti-unitary involution on $V_S^\C$
so that its real part $V_S$ is the closure of the image of $V$ in $V_S^\C$.
Moreover the alternating form $\sigma(x,y) = (S(x,y) - \overline{S}(x,y))/i$ on $V$ induces 
a continuous alternating form $\sigma_S$ on $V_S$
so that $\sigma_S([x],[y]) = \sigma(x,y)$ ($x,y \in V$), 
where $[x] = x + \ker(S + \overline{S})$ denotes the image of $x$ in $V_S^\C$.
In fact, 
in terms of the ratio operator $\mathbf{S} = (S + \overline{S})\backslash S$ on $V_S^\C$ defined by 
$([x]|\mathbf{S}[y])_S = S(x,y)$ ($x,y \in V^\C$), $\sigma_S$ is described by
$\sigma_S(\xi,\eta) = -i(\xi|(\mathbf{S} - \overline{\mathbf{S}})\eta)$ ($\xi, \eta \in V_S$).

Given an $\R$-linear map $\phi: V \to W$ or equivalently a *-linear map $V^\C \to W^\C$,
$\phi$ induces a right multiplication map $\sP(W) \to \sP(V)$ 
by $(T\phi)(v,v') = T(\phi v, \phi v')$,
which satisfies $\overline{T}\phi = \overline{T\phi}$.
In particular, the group $\text{GL}(V)$ acts on $\sP(V)$ so that it preserves complex conjugation.

\begin{Remark}
  {\small
    If $V^\C$ itself is a *-Hilbert space and $S(x,y) = (x|Ay)$ with $A$ bounded
    so that $A + \overline{A}$ has a bounded inverse,
    then $V_S^\C \ni [x] \mapsto \sqrt{A + \overline{A}} x \in V^\C$ is a unitary map and
    the operator $\mathbf{S}$ is
realized on $V^\C$ by $(A + \overline{A})^{-1/2} A (A + \overline{A})^{-1/2}$.} 
\end{Remark}

\section{Pusz-Woronowicz Functional Calculus}
Let $\alpha, \beta$ be positive (sesquilinear) forms 
on a complex vector space $H$. 
By a \textbf{representation} of the pair $\{ \alpha, \beta\}$, we shall mean 
a linear map $\iota: H \to \sH$ of $H$ into a Hilbert space $\sH$ 
together with positive (bounded) operators 
$A$, $B$ in $\sH$ such that $\iota(H)$ is dense in $\sH$ and $AB = BA$.
Such a representation is always possible:
Let $\sH$ be the Hilbert space associated to the positive form 
$\alpha + \beta$ and $\iota: H \to \sH$ be the canonical map. 
By the Riesz lemma, we have bounded operators $A$ and $B$ in $\sH$ 
representing $\alpha$ and $\beta$ respectively, which commute becuase of $A + B = 1_\sH$. 

A complex-valued Borel function $f$ on the closed first quadrant 
$[0,\infty)^2$ in $\R^2$ is called a \textbf{form function} if it is locally bounded and homogeneous 
of degree one; $f$ is bounded when restricted to a compact subset of $[0,\infty)^2$ 
and $f(rs,rt) = rf(s,t)$ for $r,s,t \geq 0$. 
Clearly $f(0,0) = 0$ and there is 
a one-to-one correspondence between form functions and bounded Borel functions on 
the unit interval $[0,1]$ by the restriction $f(t,1-t)$ ($0 \leq t \leq 1$). 
Let $\sF$ be the the vector space of form functions.

Given a form function $f$ and commuting positive operators $A$, $B$ in a Hilbert space $\sH$,
a normal operator $f(A,B)$ is defined by
\[
f(A,B)\xi = \int_{\sigma(A)\times \sigma(B)} f(s,t) e_A(ds) e_B(dt)\xi.  
\] 
Here $e_A(ds)$, $e_B(dt)$ are spectral measures of $A$, $B$,
which are projection-valued measures on their spectrum sets $\sigma(A)$, $\sigma(B)$ respectively. 

One of main results in Pusz-Woronowicz theory is the following.

\begin{Theorem}[\cite{PW1}] 
For $f \in \sF$, the sesquilinear form on $H$ defined by 
\[
\gamma(x,y) = (\iota(x)|f(A,B) \iota(y)), 
\quad 
x, y \in H
\]
does not depend on the choice of representations of $\{ \alpha, \beta\}$, 
which is the Pusz-Woronowicz functional calculus of $\alpha$, $\beta$
and will be reasonably denoted by $f(\alpha,\beta)$. 
\end{Theorem}

A real-valued function $f \in \sF$ is said to be \textbf{form concave} if
\[
(1-t)f(\alpha_0,\beta_0) + t f(\alpha_1,\beta_1) \leq f((1-t)\alpha_0 + t\alpha_1, (1-t)\beta_0 + t\beta_1)
\]
for $0 \leq t \leq 1$ and
four positive forms $\alpha_j$, $\beta_j$ ($j=0, 1$) defined on a common vector space.
The condition then turns out to be equivalent to requiring that
\[
f(\alpha,\beta)\phi \leq f(\alpha\phi,\beta\phi) 
\]
for any linear map $\phi: H \to K$ and positive forms $\alpha$, $\beta$ on $K$ (\cite{PW2}).
Here, for a sesquilinear form $\gamma$ on $K$, $\gamma\phi$ is a sesquilinear form on
$H$ defined by $(\gamma\phi)(x,y) = \gamma(\phi x,\phi y)$ ($x,y \in H$).

When $f \in \sF$ takes values in $[0,\infty)$,
the concavity is closely related to the notion of operator mean in
Kubo-Ando theory and equivalent to the operator monotonicity of the function $f(1,t)$ (\cite{Fu}, \cite{KA}).
If this is the case, thanks to the integral representation of operator-monotone functions,
we have an expression
\[
f(s,t) = \int_{[0,\infty]} (1+\lambda) \frac{st}{\lambda s + t}\, \mu(d\lambda)
\]
with $\mu$ a finite positive measure on $[0,\infty]$ so that 
\[
  f(\alpha,\beta) 
  = \int_{[0,\infty]} (1+\lambda) \frac{\alpha\beta}{\lambda \alpha + \beta}\, \mu(d\lambda). 
\]


Recall that an operator-monotone function $f:[0,\infty) \to [0,\infty)$ is characterized by the property
that it is analytically extended to a holomorphic function $\R + i(0,\infty) \to \R + i[0,\infty)$.
The power function $t^p$ ($p>0$) is therefore operator-monotone if and only if $0 < p \leq 1$
with the accompanied form concave function $f(s,t) = s^{1-p} t^p$.

\section{A Flow on Positive forms}
For a real number $r>0$, consider a homogeneous continuous function 
\[
  f_r(s,t) = \begin{cases} \frac{s^r(s - t)}{s^r - t^r} & (s \not= t)\\
    t/r &(s=t)
  \end{cases}, 
\]
which is strictly positive on $(0,\infty)\times (0,\infty)$ so that its sectional function
\[
  f_r(s,1-s) = f(x) = \frac{1-x}{(1+x)(1-x^r)}, 
  \quad
  x = \frac{1-s}{s} \iff s = \frac{1}{1+x} 
\]
is globally bounded on $x \geq 0$ with $f(0) = 1$, $f(\infty) = 0$
and pointwise decreasing in $r>0$. 

Note that $f_1(s,t) = s$ and 
$f_r(s,t) - s = (st^r - s^rt)/(s^r - t^r)$ is symmetric under the exchange $s \leftrightarrow t$. 

Given a positive form $S$ on $V^\C$ and $r>0$, introduce a positive form $S^{(r)}$ by 
\[
  S^{(r)} = f_r(S,\overline{S}) = \frac{S^r(S - \overline{S})}{S^r - {\overline S}^r}
  = S + \frac{S{\overline S}^r - S^r {\overline S}}{S^r - {\overline S}^r}.  
\]

We first observe that $S^{(r)} + \overline{S^{(r)}}$ is equivalent to $S + \overline{S}$ as a positive form.
In fact, it is a result of functional calculus based on a function
\[
  \frac{s^{r+1} - s^r\overline{s}}{s^r - {\overline s}^r}
  + \frac{{\overline s}^{r+1} - {\overline s}^rs}{{\overline s}^r - s^r}
  = s \frac{1 - (\overline{s}/s) + (\overline{s}/s)^r - (\overline{s}/s)^{r+1}}{1 - (\overline{s}/s)^r}  
\]
of $s$ and $\overline s$, whose restriction to $\overline s = 1 - s$ is described by the continuous function 
\[
h(x) =  \frac{(1-x)(1+x^r)}{(1+x)(1-x^r)}
\]
of a new parameter $x = (1-s)/s \in [0,\infty]$. 
Clearly $h$ is strictly positive and satisfies $h(0) = 1$, $h(\infty) = 1$,
whence $1/h$ as well as $h$ is bounded.

Let $T = S^{(a)}$ with $a>0$. We claim that $T^{(b)} = S^{(ab)}$ for $b>0$, i.e.,
the non-linear operation $S \mapsto S^{(r)}$ defines a flow on the convex set of positive forms on $V^\C$,
which is referred to as a \textbf{scaling flow}.

To see this, represent $S$ and $\overline{S}$ by commuting positive operators, which can be identified
with multiplication operators by parameters $s$ and $\overline s$ via a joint-spectral decomposition
in such a way that the complex conjugate of $s$ is given by $\overline s$ as a multiplication operator.
Then $T$ and $\overline{T}$ are realized as multiplication operators by the functions
\[
  t = \frac{s^a(s-\overline{s})}{s^a - {\overline s}^a},
  \quad 
  \overline{t} = \frac{{\overline s}^a(s-\overline{s})}{s^a - {\overline s}^a}
\]
respectively.

Now $T^{(b)}$ is represented as a multiplication  operator of the function
\[
  \frac{t^b(t-\overline{t})}{t^b - {\overline t}^b}, 
 \]
which, with a straightforward rearrangement, turns out to be equal to 
\[
  \frac{s^{ab}(s-\overline{s})}{s^{ab} - {\overline s}^{ab}},
\]
showing $T^{(b)} = S^{(ab)}$.

\begin{Remark}
  When $S = \overline{S}$ ($\sigma \equiv 0$),
  the flow $S^{(r)}$ is simply $r^{-1} S$. 
\end{Remark}

\begin{Remark}
  {\small
If $S$ is given by a bounded positive operator $A$ on a *-Hilbert space $V^\C$
with $A + \overline{A}$ having a bounded inverse,
then $S^{(r)}$ is realized on $V^\C$ by the operator
\[
\frac{((A + \overline{A})^{-1/2} A(A + \overline{A})^{-1/2})^r
  (A + \overline{A})^{-1/2}(A - \overline{A}) (A + \overline{A})^{-1/2}}
  {((A + \overline{A})^{-1/2} A(A + \overline{A})^{-1/2})^r
  -((A + \overline{A})^{-1/2} \overline{A}(A + \overline{A})^{-1/2})^r}. 
\]
}
\end{Remark}

Here are some simple facts on scaling flows:
\begin{enumerate}
\item
  As already observed, positive forms in a scaling orbit have equivalent real parts.
\item
  Since $S^{(r)} - S = \frac{S{\overline S}^r - S^r\overline{S}}{S^r - {\overline S}^r}$ is real,
the scaling flow preserves the imaginary part $(S - \overline{S})/i$,
which corresponds to an alternating form on $V$.
In other words, given an alternating form $\sigma$ on a real vector space $V$,
if we define a convex subset of
$\sP(V)$ by $\sP(V,\sigma) = \{ S \in \sP(V); S - \overline{S} = i\sigma\}$,
then the scaling flow leaves $\sP(V,\sigma)$ invariant.
\item
Let $V = V_1 \oplus \cdots \oplus V_n$ and $S = S_1 \oplus \cdots \oplus S_n$ with
$S_j \in \sP(V_j)$. Then $S^{(r)} = S_1^{(r)} \oplus \cdots \oplus S_n^{(r)}$.
\item
Let $\phi: V \to W$ be a $\R$-linear map or equivalently a *-linear map of $V^\C$ into $W^\C$.
For a positive form $T \in \sP(W)$,
$(T\phi)^{(r)}$ is generally different from $T^{(r)}\phi$
but here is a simple criterion for their coincidence:
Let $W_T^\C$ be the Hilbert space associated to
$T + \overline{T}$ and let $E$ be the projection to the closure of $\phi(V^\C)$ in $W_T^\C$.
If the ratio operator $(T + \overline{T})\backslash T$ commutes with $E$,
then $(T\phi)^{(r)} = T^{(r)}\phi$ ($r>0$).

This is the case when $\phi$ is an isomorphism.
In particular, the natural action of $\Aut(V,\sigma)$ on $\sP(V,\sigma)$ commutes with the operation of
taking scaling flow.

Another relevant case is the canonical map $\phi:V \to V_S$ ($S \in \sP(V)$) with $T \in \sP(V_S)$
specified by $T(\phi(x),\phi(y)) = S(x,y)$ and continuity.
Furthermore, 
under the orthogonal decomposition $V_S = V_{S = \overline{S}} \oplus V_{S \not= \overline{S}}$ according
to the point spectrum $(S + \overline{S})\backslash S = 1/2$ and the remaining,
the flow $T^{(r)}$ is split into two parts so that it is represented on $V_{S = \overline{S}}$ by
a scalar operator $1/2r$. When factored out this trivial part, we are reduced to
the case that $1/2$ is not an eigenvalue of $(S + \overline{S})\backslash S$.
A positive form $S$ fulfilling this condition is said to be \textbf{center-free}.

The center-freeness of $S$ is equivalent to 
the non-degeneracy of the completed alternating form $\sigma_S$ on $V_S$. 


Recall that each $S \in \sP(V,\sigma)$ gives rise to a free state $\varphi_S$ on the CCR-algebra
of $(V_S,\sigma_S)$ 
and the center $Z$ of the GNS representation $\pi_S$ of $\varphi_S$ is generated by $\pi_S(\ker\sigma_S)$. 
The center-freeness of $S$ is then equivalent to the triviality of $Z$.
\end{enumerate}

\begin{Proposition}
  A positive form $S$ is extremal in the convex set $\sP(V,\sigma)$
  if and only if $\sqrt{S\overline{S}} = 0$. 
\end{Proposition}

\begin{proof}
With the notation at the end of \S 1, 
$\sqrt{S\overline{S}}(x,y) = ([x],\sqrt{\mathbf{S}\overline{\mathbf{S}}} [y])_S$. 
If $\int_{[0,1]} s\, e(ds)$ denotes the spectral representation of $\mathbf{S}$,
then $\overline{\mathbf{S}} = \int_{[0,1]} (1-s) e(ds)$ and 
$\sqrt{\mathbf{S}\overline{\mathbf{S}}} = \int_{[0,1]} \sqrt{s(1-s)} e(ds) \not= 0$ implies 
$e([\epsilon,1-\epsilon]) \not= 0$ for some $0 < \epsilon < 1/2$. 
To see the non-extremality of $S$,
we seek for a non-zero hermitian operator $h = \overline{h}$ satisfying $\mathbf{S}\pm h \geq 0$
so that $([x],(\mathbf{S}\pm h)[y])_S$ are covariance forms and $S$ is their average. 
Assume that $h = \int f(s)\, e(ds)$ with $f$ bounded and real.
Then the condition takes the form  
$f(s) = f(1 - s)$ and $-s \leq f(s) \leq s$, which is satisfied by 
the choice $f = \epsilon 1_{[\epsilon,1-\epsilon]}$ for example. 

Conversely, assume that $\sqrt{S\overline{S}} = 0$ and $S = (R + T)/2$ with 
$R, T \in \sP(V,\sigma)$. 
Then $\mathbf{S} + \overline{\mathbf{S}} = 1$ and $\sqrt{\mathbf{S}\overline{\mathbf{S}}} = 0$
show that $\mathbf{S}$ is a projection $E$ in $V_S^\C$. 
Since $R/2 \leq S$ and $T/2 \leq S$, positive forms $R$ and $T$ 
are represented by bounded positive operators on $V_S^\C$;
we can find a hermitian operator $h$ 
in $V_S^\C$ so that $R(x,y) = ([x],(E - h)[y])_S$ and $T(x,y) = ([x],(E + h)[y])_S$. 
From $E - h - \overline{E - h} = E - \overline{E}$,
$h = \overline{h}$. 
By positivity of $R$ and $T$, we have $-E \leq h \leq E$ 
and then
\[
-\overline{E} E \overline{E} \leq \overline{E} h \overline{E} 
\leq \overline{E} E \overline{E} 
\]
implies $\overline{E} h \overline{E} = 0$ and, by taking the conjugation, $EhE = 0$. 
Thus, if we set $C = \overline{E} h E$, the operator $T$ has the matrix representation 
\[
\begin{pmatrix}
1 & C\\
C^* & 0
\end{pmatrix}
\]
with respect to the decomposition $V_S^\C = EV_S^\C \oplus \overline{E}V_S^\C$,
which is positive only if $C = 0$, showing $h = C + C^* = 0$. 
\end{proof} 

\begin{Remark}
{\small  There is no extremal point in $\sP(V)$ other than $0$.} 
\end{Remark}

\begin{Theorem}\label{freeze}~ 
  \begin{enumerate}
  \item
    A positive form $S \in \sP(V,\sigma)$ is a fixed point of the scaling flow
    if and only if $S$ is extremal in $\sP(V,\sigma)$.
  \item
  Each flow $S^{(r)}$ converges pointwise to an extremal point $S^{(\infty)}$,
where $S^{(\infty)}$ is a positive form defined by 
$S^{(\infty)}(x,y) = ([x],(2\mathbf{S}-1)_+ [y])_S$ 
with $(2\mathbf{S}-1)_+$ denoting the cut of $2\mathbf{S} - 1$ by the spectral range $(0,1]$
and satisfies $\sqrt{S^{(\infty)} \overline{S^{(\infty)}}} = 0$.
\end{enumerate}
\end{Theorem}

\begin{proof}
(i) Suppose that $S^{(r)} = S$ for some $r \not= 1$. In terms of the spectral measure $e(ds)$ 
  of $\mathbf{S}$, this is equivalent to
  \[
    \int_{[0,1]} \frac{s(1-s)^r - s^r(1-s)}{s^r - (1-s)^r}\, e(ds) = 0.
\]
  Since the integrand is $(x^r-x )/((1+x)(1-x^r))$ as a function of $x = (1-s)/s$, it vanishes if and only if
  the spectral measure is supported by $\{0,1\}$, i.e., $\mathbf{S}\overline{\mathbf{S}} = 0$.
  
  (ii)  With the same spectral expression, $S^{(r)}$ is represented by a positive operator
  \[
    \int_{[0,1]} \frac{s^r(2s-1)}{s^r - (1-s)^r}\, e(ds)
  \]
  and the assertion follows from the fact that the integrand is decreasing in $r>0$ and 
  \[
\lim_{r \to \infty} 
\frac{s^{r+1} - s^r(1-s)}
{s^r - (1-s)^r} 
= 
\begin{cases}
0 &\text{if $0 \leq s \leq 1/2$,}\\
2s - 1 &\text{if $1/2 < s \leq 1$.}
\end{cases}
\]
\end{proof}


\begin{Remark}
  {\small
    The positive function $f_r(1,t) = \frac{1-t}{1-t^r}$ on $[0,\infty)$ is operator-monotone
    if and only if $0 < r \leq 1$.
    Consequently, $f_r(s,t)$ is form concave if and only if $0 < r \leq 1$.
  
  Related to this fact,
  the symmetric homogeneous function $g_r(s,t) = \frac{r}{1-r} \frac{st^r - s^rt}{s^r - t^r}$
  (the value for $r=1$ or $s=t$ being assigned by continuity)
  may arouse some interest, whose accompanied function $g_r(1,t)$ ($t \geq 0$) 
  is operator-monotone if and only if $0 < r \leq 2$.
  In other words, $g_r(s,t)$ is form concave if and only if $0 < r \leq 2$. Note that 
  \[
    g_1(s,t) = st \frac{\log s - \log t}{s-t},
    \quad 
    g_r(t,t) = t.
  \]
  See \cite{NW} for more information on operator-monotone functions of this type.

  In \cite{HK}, mean operation is assigned to continuous form functions
  for which $f(s,1)$ and $f(1,t)$ are increasing.
  We notice that both $g_r(s,1) = g_r(1,s)$ and $f_r(s,1)$ are always increasing, whereas
  $f_r(1,t)$ is increasing if and only if $0 < r < 1$ ($f_r(1,t)$ being decreasing for $r>1$).   
  }
\end{Remark}


\section{Free States and Scaled Dynamics}
Given an alternating vector space $(V,\sigma)$, let $C(V,\sigma)$ be the C*-algebra generated by
the Weyl commutation relations: $C(V,\sigma)$ is universally generated by a unitary family $\{ e^{ix}\}_{x \in V}$
subject to the condition (Weyl commutation relations) 
$e^{ix} e^{iy} = e^{-i\sigma(x,y)/2} e^{i(x+y)}$ ($x,y \in V$).
Thanks to the universality, a linear map $\phi$ of $(V,\sigma)$ into another
alternating vector space $(V',\sigma')$ satisfying $\sigma'\phi = \sigma$ gives rise to a *-homomorphism
$C(\phi)$ of $C(V,\sigma)$ into $C(V',\sigma')$, which turns out to be injective for an injective $\phi$.
In particular, we have a natural group monomorphism of $\Aut(V,\sigma)$ into $\Aut(C(V,\sigma))$.

A big fault in $C(V,\sigma)$ is that, 
although $e^{itx}$ ($t \in \R$) is a one-parameter group of unitaries in $C(V,\sigma)$,
no continuity on $t \in \R$ is guaranteed in $e^{itx}$.
We compensate this by assuming continuity on representations and states instead. 

Each $S \in \sP(V,\sigma)$ gives a state $\varphi_S$ on $C(V,\sigma)$ specified by
$\varphi_S(e^{ix}) = e^{-S(x,x)/2}$,
which is referred to as a \textbf{free state} of \textbf{covariance form} $S$. 
Note that $\varphi_S(e^{itx})$ is continuous in $t \in \R$ and
the associated GNS representation $\pi_S$ enjoys the continuity of $\pi_S(e^{itx})$ on $t \in \R$.

If a linear map $\phi: V \to V'$ relates $S \in \sP(V,\sigma)$ and $S' \in \sP(V',\sigma')$ by
$S'\phi = S$, then $\varphi_S$ comes from $\varphi_{S'}$ through $C(\phi)$. 
This is the case for the canonical map $(V,\sigma) \to (V_S,\sigma_S)$ associated with $S \in \sP(V,\sigma)$
and a free state $\varphi_S$ on $C(V,\sigma)$ can be analysed through $C(V_S,\sigma_S)$. 

Recall that a state $\varphi$ on a C*-algebra $A$
satisfies the KMS condition with respect to a one-parameter automorphism group $\theta_t$ of $A$ if
$\varphi(a\theta_t(b))$ ($a,b \in A$)
is a continuous function of $t \in \R$ and analytically continued to the strip region
$\{ z \in \C; -1 \leq \text{Im}(z) \leq 0 \}$ so that
$\varphi(a\theta_{t-i}(b)) = \varphi(\theta_t(b)a)$ ($t \in \R$).
A state $\varphi$ is called a KMS state if it satisfies the KMS condition with respect to some $(\theta_t)$. 
For a KMS state $\varphi$, a one-parameter automorphism group fulfilling the KMS condition
is substantially determined by $\varphi$ (Takesaki's result \cite[Theorem~5.3.10]{BR2}), 
whereas different states may share a common $(\theta_t)$ to meet the KMS condition. 

Now we assume that $S$ is \textbf{non-boundary}
in the sense that the ratio operator $\mathbf{S}$ has a trivial kernel.
Since $\overline{\mathbf{S}} = 1 - \mathbf{S}$, the condition is equivalent to requiring
$\ker(1- \mathbf{S}) = \{ 0\}$. A self-adjoint operator $h$ on $V_S^\C$ is now introduced by
\[
  h = \log\frac{\overline{\mathbf{S}}}{\mathbf{S}} \iff \mathbf{S} = \frac{1}{1+ e^h}, 
\]
which satisfies $\overline{h} = - h$
($\overline{h}$ being the complex conjugate of $h$ and not the closure of $h$).
The associated one-parameter group of unitaries $(e^{ith})$ in $V_S^\C$ preserves 
the alternating form $\sigma_S$, whence it induces a one-parameter group $(\theta_t)$ of
*-automorphisms of $C(V_S,\sigma_S)$.
It is then a well-known fact (cf.~\cite[Example~5.3.2]{BR2}) that the free state $\varphi_S$ on $C(V_S,\sigma_S)$
satisfies the KMS condition with respect to $(\theta_t)$.

According to its physical interpretation, the KMS condition for the scaled automorphism group 
$(\theta_{rt})$ describes equilibrium states of inverse temparature multiplied by the factor $r>0$ compared to
the original one $(\theta_t)$.

Thus a free state of a generic covariance form $T$ satisfies the KMS condition for $(\theta_{rt})$ if 
$S + \overline{S} \sim T + \overline{T}$ and $\overline{\mathbf{T}}/\mathbf{T} = e^{rh}$,
which we shall now rewrite in terms of $S$. 
Notice here that $V_S = V_T$ and $\sigma_S = \sigma_T$. 
By passing to the completed space $V_S = V_T$, we may assume that $V = V_S = V_T$ and
$\sigma = \sigma_S = \sigma_T$ for this purpose.

In the extreme case $\sigma = 0$,
both $S = \overline{S}$ and $T = \overline{T}$ give rise to the trivial automorphism group,
losing the meaning of scaling effect.
To eliminate this somewhat degenerate situation, we further assume that the common $\sigma$ is non-degenerate.
In view of $\frac{S - \overline{S}}{S + \overline{S}} = \mathbf{S} - \overline{\mathbf{S}} = 2\mathbf{S} - 1$,
non-degeneracy of $\sigma$ is equivalent to $\ker(2\mathbf{S} - 1) = \{ 0\}$,
i.e., $1/2$ is not an eigenvalue of $\mathbf{S}$ and similarly for $\mathbf{T}$. 

From the condition $e^{rh} = \overline{\mathbf{T}}/\mathbf{T}$,
the positive self-adjoint operator $e^{rh}$ commutes with $\mathbf{T}$ and, 
for $x \in V^\C$ and $y$ in the domain of $e^{rh}$,
\[
  T(x,e^{rh}y) = \overline{T}(x,y) = T(x,y) - S(x,y) + \overline{S}(x,y). 
\]
Thus the ratio operator $(S + \overline{S}) \backslash T$ is given by 
\[
  \frac{1}{1-e^{rh}} \frac{S - \overline{S}}{S + \overline{S}}
  = \frac{1}{1 - \left( \frac{1-\mathbf{S}}{\mathbf S} \right)^r}(2\mathbf{S} - 1)
  = \frac{\mathbf{S}^r}{\mathbf{S}^r - (1-\mathbf{S})^r} (\mathbf{S} - (1 - \mathbf{S})) 
\]
and we conclude that 
\[
  T(x,y) = (x, \frac{\mathbf{S}^r}{\mathbf{S}^r - \overline{\mathbf{S}}^r} (\mathbf{S} - \overline{\mathbf{S}}) y)_S
  = S^{(r)}(x,y). 
\]
Here notice that both $S^{(r)}$ and $T$ are continuous on Hilbertian space $V_S^\C = V_T^\C$.

\begin{Theorem}
  For a non-boundary $S \in \sP(V,\sigma)$, so is $S^{(r)}$  and,
  if the associated KMS automorphism group of $C(V_S,\sigma_S)$ is denoted by $\theta_t^{(r)}$,
  then $\theta^{(r)}_t = \theta_{rt}$ for $t \in \R$ and $r>0$. 
\end{Theorem}

\begin{proof}
  With the decomposition, $V_S = V_{S = \overline{S}} \oplus V_{S \not= \overline{S}}$,
  the KMS condition is satisfied
  in such a way that $e^{ix}$ ($x \in V_{S = \overline{S}}$) are left invariant
  under the automorphism group. 
\end{proof}

Now the physical meaning of Theorem~\ref{freeze} is clear; 
free states of covariance forms $S^{(r)}$ and $S$ describe equilibrium states of
the same dynamics with their temperatures related by scaling of a factor $r>0$ and the limit $r \to \infty$ freezes
the free state of covariance form $S$ to zero temperature,
resulting in a Fock state whose covariance form $S^{(\infty)}$ is described by the spectral projection
of $\mathbf{S}$ for the spectral range $(1/2,1]$. Notice that, if $\ker(2\mathbf{S} - 1) = \{ 0\}$
and $S$ is not a Hilbert-Schmidt perturbation of Fock form,
free states of covariance form $S^{(r)}$ generates disjoint factors for different $r$'s
(see \cite{AY} for example).

\section{Powers of Density Operators}


When $V$ is finite-dimensional, we can regularize $C(V,\sigma)$ by integration to get
a C*-algebra $C^*(V,\sigma)$ whose *-representations are exactly continuous representations of $C(V,\sigma)$.

Representations of $C^*(V,\sigma)$ are simple enough to see 
$C^*(V,\sigma) \cong C_0((\ker \sigma)^*)\otimes \sC(\sH)$,
where the dual vector space $(\ker \sigma)^*$ is considered to be a locally compact space,
$\sH$ is an $L^2$-space of euclidean dimension $\dim(V/\ker\sigma)/2$
and $\sC(\sH)$ denotes the compact operator algebra in a Hilbert space $\sH$. 
When $\sigma$ is non-degenerate,
this gives the uniqueness theorem of Stone-von Neumann on CCR
and in fact these are essentially equivalent. States of $C^*(V,\sigma)$ are
therefore identified with density operator-valued probability measures on $(\ker \sigma)^*$.

Relying on the method developed in \cite{gqfs}, 
we shall here describe free states in terms of density operators and
give an explicit formula for powers of free density operators. 

\subsection{}
In the following, we fix a Lebesgue measure $dx$ in $V$ once for all
and make $L^1(V)$ into a *-algebra, which is denoted by $L^1(V,\sigma)$ to indicate the dependence on
$\sigma$, in such a way that a continuous representation $\pi(e^{ix})$ of
the Weyl commutation relations gives rise to a *-representation of $L^1(V)$ by
\[
  \pi(f) = \int_V f(x) \pi(e^{ix})\, dx. 
\]
It is immediate to write down the explicit form: 
\[
  (fg)(x) = \int_V f(y) g(x-y) e^{i\sigma(x,y)/2}\,dy,
  \quad
f^*(x) = \overline{f(x)}
\]
for $f, g \in L^1(V)$.
Then $C^*(V,\sigma)$ is just the enveloping C*-algebra of $L^1(V,\sigma)$ and we also use the notation
$\int f(x) e^{ix}\, dx$ ($f \in L^1(V)$) to stand for an element in $C^*(V,\sigma)$ in view of the fact that
a *-representation $\pi$ of $C^*(V,\sigma)$ is in a one-to-one correspondece with a continuous representation
$\pi(e^{ix})$ of Weyl unitaries by the relation $\pi(f) = \int f(x) \pi(e^{ix})\, dx$. 
Likewise, a positive functional $\varphi$ of $C(V,\sigma)$ for which the characteristic function
$\varphi(e^{ix})$ is continuous in $x \in V$ is in a one-to-one correspondence with a positive functional
of $C^*(V,\sigma)$ by the relation
$\varphi(f) = \int f(x) \varphi(e^{ix})\, dx$ ($f \in L^1(V)$).
Particularly, a free state $\varphi_S$ of $C(V,\sigma)$ corresponds to a state $\varphi$
of $C^*(V,\sigma)$ described by $\varphi(f) = \int e^{-S(x,x)/2} f(x)\, dx$ ($f \in L^1(V)$).
In what follows, $\varphi$ is also denoted by
$\varphi_S$ and referred to as a free state of $C^*(V,\sigma)$. 

A functional $\tau$ is now introduced on a dense *-ideal $L(V,\sigma) = C_0(V) \cap L^1(V)$ of $L^1(V,\sigma)$ by
$\tau(f) = f(0)$ ($f \in L(V,\sigma)$), which is positive and tracial from 
\[
  \tau(f^*f) = \int_V |f(x)|^2\, dx = \tau(ff^*).
\]
The associated GNS representation $\pi$ of $L^1(V,\sigma)$ is then realized on $L^2(V)$ by
\[
  (\pi(f)\xi)(x) = \int_V f(y) \xi(x-y) e^{i\sigma(x,y)/2}\,dy
\]
and gives rise to a faithful imbedding of $C^*(V,\sigma)$ into $\sB(L^2(V))$.

\begin{Remark}
  {\small
  When $\sigma$ is non-degenerate, $C^*(V,\sigma)) \cong \sC(\sH)$ so that $L^2(V) \cong \sH\otimes \sH^*$
  and $\tau$ turns out to be proportional to the ordinary trace by a factor 
  $(2\pi)^{-\dim V/2} \text{tr}$ thanks to von Neumann's formula on a minimal projection in
  $C^*(V,\sigma)$ (\cite{vN}).}
\end{Remark}

For $S \in \sP(V,\sigma)$ with $S + \overline{S}$ positive definite, $\rho_S(x) = e^{-S(x,x)/2}$ belongs to
$L(V,\sigma)$ and satisfies $\varphi_S(f) = \tau(\rho_S f)$ for $f \in L^1(V)$.
In other words, $\rho_S$ is a density operator as an element of $C^*(V,\sigma)$ which represents $\varphi_S$
with respect to $\tau$.

By the positivity of $\varphi_S$, $\rho_S$ is a positive element of $C^*(V,\sigma)$ and 
  its positive power $\rho_S^r$ has a meaning as a positive element as well.
  We shall show that $\rho_S^r$ is proportional to $\rho_{S^{(r)}}$ with an explicit formula
  for the proportional constant $\tau(\rho_S^r)$.
  Notice here that
  $S^{(r)} + \overline{S^{(r)}}$ is equivalent to $S + \overline{S}$ and therefore $\rho_{S^{(r)}} \in L(V,\sigma)$.

  We recall the following standard fact in linear algebra.
  
  \begin{Lemma}\label{basis}
    We can find a basis $(h_i,p_j,q_j)$ of $V$ which is orthonormal with respect to $S + \overline{S}$ and
    satisfies
    \[
      \mathbf{S}h_i = \frac{1}{2}h_i,
      \quad
      \mathbf{S}p_j = \frac{1}{2}p_j + i\mu_jq_j,
      \quad
      \mathbf{S}q_j = \frac{1}{2}q_j - i\mu_jp_j. 
    \]
    Here $0 < \mu_j \leq 1/2$ and
    $\frac{1}{2} - \mu_j$'s are eigenvalues (including multiplicity) of $\mathbf{S}$
    in the range $[0,1/2)$. 
    Note that the alternating form satisfies $\sigma(q_j,p_k) = 2\mu_j\delta_{j,k}$ and $S$ is non-boundary
    if and only if $\mu_j < 1/2$ for every $j$. 
  \end{Lemma}

  Since $\varphi_S$ is factored as a product state according to the spectral decomposiotn of $\mathbf{S}$
  in Lemma~\ref{basis}, the problem is reduced to the case $V = \R$ ($\sigma = 0$)
  or $V = \R^2$ ($\sigma$ being non-degenerate).
  %
  %
  %
  %
  %
  %
%
%
%
%
%

\subsection{}
We now focus on a single component case $V = \R p + \R q$ with $\sigma(q,p) = 2\mu$ ($0 < \mu \leq 1/2$).
Then, in terms of the linear coordinates $xp + yq$ ($x,y \in \R$),
$\rho_S(x,y) = e^{-(x^2+y^2)/4}$ and the Liouville measure is given by $2\mu dxdy$. 

To get an explicit formula for $\rho_S^r$,
we assume $\mu < 1/2$ for the moment and start with the Gaussian integral formula (cf.~\cite{gqfs} \S 4.2)
\[
e^{-\mu(x^2+y^2)/2a}* e^{-\mu(s^2+y^2)/2b} 
= \frac{4\pi mab}{\mu(a+b)} e^{-\mu(s^2+y^2)/2a*b}, 
\quad 
a*b = \frac{a+b}{1+ab}
\]
in $L^1(V,\sigma)$ (the convolution product being taken with respect to the measure $2mdxdy$), 
which suggests putting $a = \tanh\theta$ ($\theta>0$) with $\theta$ an additive parameter. 

To apply this formula to $\rho_S$, we choose $a = 2\mu$ and set 
$2\mu = \tanh\theta$ ($\theta > 0$) with the restriction $\mu < 1/2$ maintained.  
The positive element $\rho_S^r$ is then realized by
a Gaussian function of the variance parameter $\tanh(r\theta)$:
\[
\rho_S^r(x,y) = w(r)
\exp\left(
- \frac{\tanh\theta}{4\tanh(r\theta)} (x^2+y^2) 
\right) 
\]
with $w(r) = \tau(\rho_S^r)$ a positive constant depending continuously on $r>0$. 

To determine $w(r)$, we look at the relation
$\rho_S^r\rho_S^{r'} = \rho_S^{r+r'}$, i.e., 
\[
w(r+r') = \frac{4\pi m}{\mu} 
\frac{\tanh(r\theta) \tanh(r'\theta)}
{\tanh(r\theta) + \tanh(r'\theta)}
w(r)w(r').
\]
By the identity 
\[
\frac{1}{\sinh(\alpha + \beta)} 
= \frac{\tanh\alpha \tanh\beta}{\tanh\alpha + \tanh\beta} 
\frac{1}{\sinh\alpha} \frac{1}{\sinh\beta},
\]
$(4\pi m/\mu) \sinh(r\theta) w(r)$ is multiplicative in $r$ and we conclude that 
\[
w(r) = \left(\frac{4\pi m}{\mu}\right)^{r-1} \frac{(\sinh \theta)^r}{\sinh(r\theta)}
\]
in view of $w(1) = 1$.
Remark here that the case $\mu = 1/2$ is covered by 
these formulas if we take the limit $\theta \to \infty$. 
In fact, with respect to the Liouville measure ($m = \mu$), 
\[
w(r) = \lim_{\theta \to \infty} 
\frac{(4\pi \sinh \theta)^r}{4\pi \sinh(r\theta)}
= (2\pi)^{r-1}
\]
matches with the correct formula 
$\rho_S^r = (2\pi)^{(r-1)\dim V/2} \rho_S$ for an extremal $S$.  
  
All these local results are now written in a global form. 
Introduce an operator $H$ on $V^\C$ so that $\tanh H$ is represented by the matrix 
$\begin{pmatrix}
0 & 2i\mu_j\\
-2i\mu_j & 0  
\end{pmatrix}$ on the subspace $\C p_j + \C q_j$, i.e., 
\[
\mathbf{S} - \overline{\mathbf{S}} = \tanh H.
\]
Here we allow $H$ to take eigenvalues $\pm\infty$. 
Since $\tanh$ is an odd function and the left hand side 
is purely imaginary, we see $\overline{H} = -H$.  
From the relations
\[
\mathbf{S} = \frac{1 + \tanh H}{2} = \frac{e^H}{e^H + e^{-H}}, 
\quad 
\overline{\mathbf{S}} = \frac{1 + \tanh(-H)}{2} 
= \frac{e^{-H}}{e^H + e^{-H}},
\]
we obtain the expression 
$2H = \log({\mathbf{S}}{\overline{\mathbf{S}}}^{-1})$. 
In terms of $H$, we have 
\[
S^{(r)}(x,y) = 
\frac{1}{2}(x, \frac{\tanh H}{\tanh rH} (1 + \tanh rH)y)_S. 
\]
The real quadratic form $S^{(r)}(x,x) - S(x,x)$ is then 
represented by the real operator 
\[
\frac{1}{2} \frac{\tanh H - \tanh rH}{\tanh rH} 
= \frac{e^{(1-r)H} - e^{(r-1)H}}
{(e^H + e^{-H}) (e^{rH} - e^{-rH})}.
\]
Replacing $e^H$ with ${\mathbf{S}}^{1/2}
{\overline{\mathbf{S}}}^{-1/2}$, this is further reduced to  
\[
\frac{{\mathbf{S}} {\overline{\mathbf{S}}}^r 
- {\mathbf{S}}^r\overline{\mathbf{S}}}
{(\mathbf{S} + \overline{\mathbf{S}}) ({\mathbf{S}}^r - 
{\overline{\mathbf{S}}}^r)} 
= 
\frac{{\mathbf{S}} {\overline{\mathbf{S}}}^r 
- {\mathbf{S}}^r\overline{\mathbf{S}}}
{{\mathbf{S}}^r - 
{\overline{\mathbf{S}}}^r}. 
\]
Thus
\[
S^{(r)}(x,y) = S(x,y) + 
\frac{S{\overline S}^r - S^r {\overline S}}
{S^r - {\overline S}^r}
(x,y)
\]
for $x, y \in V^\C$ and we are reduced to the scaling flow on $S$.


In a similar way, we have a coordinate-free expression for $\tau(\rho_S^r)$,
this time depending on the choice of a Lebesgue measure in $V$.
To be explicit, we first work with the Liouville measure ($m = \mu$) to get 
\begin{align*}
\tau(\rho_S^r) &= 
\det\left(
\frac{|4\pi \sinh H|^r}{4\pi |\sinh(rH)|}
\right)^{1/2}\\
&= (2\pi)^{(r-1)\dim V/2} 
\det\left(
\frac{(e^{2H} + e^{-2H} - 2)^r}{e^{2rH} + e^{-2rH} - 2}
\right)^{1/4}\\
&= (2\pi)^{(r-1)\dim V/2} 
\det\left(
\frac{|S-\overline{S}|^r}{|S^r - {\overline S}^r|}
\right)^{1/2}. 
\end{align*}
Note that the operator 
$|S - \overline S|^r/|S^r - {\overline S}^r|$ on $V^\C$ 
is well-defined as a functional calculus of homogeneous degree $0$. 

\begin{Theorem} Let $(V,\sigma)$ be a finite-dimensional 
symplectic vector space and $S \in \sP(V,\sigma)$.  
Then, with the Liouville measure as a reference,
$\rho_S^r$ ($r>0$) belongs to $L(V,\sigma) = C_0(V) \cap L^1(V)$ and is given by
\[
  \rho_S^r = w(r) \rho_{S^{(r)}}
    \quad
    \text{with}
    \quad 
w(r) = 
(2\pi)^{(r-1)\dim V/2} 
\det\left(
\frac{|S-\overline{S}|^r}{|S^r - {\overline S}^r|}
\right)^{1/2}.  
\]
\end{Theorem}

Next we work with the Euclidean measure associated to the inner product $S + \overline{S}$, 
which is non-degenerate due to the non-degeneracy of $\sigma$.
In local coordinates, this corresponds to the choice $m = 1/2$ and we have 
\begin{align*}
  \tau(\rho_S^r)
  &= \prod_j \left(\frac{2\pi}{\mu_j}\right)^{r-1} \frac{(\sinh \theta_j)^r}{\sinh(r\theta_j)}\\
  &= (2\pi)^{(r-1) \dim V /2}
    \det\left(\frac{|S - \overline{S}|}{S + \overline{S}}\right)^{(1-r)/2} 
    \det\left(
\frac{|S-\overline{S}|^r}{|S^r - {\overline S}^r|}
    \right)^{1/2}\\
  &= (2\pi)^{(r-1) \dim V /2}
    \det\left(
\frac{(S+\overline{S})^{r-1} |S-\overline{S}|}{|S^r - {\overline S}^r|}
    \right)^{1/2}.  
\end{align*}

Keeping the non-degeneracy of $S + \overline{S}$,
we finally relax the alternating form $\sigma$ to be degenerate. 
In local coordinates, the degenerate part of $\sigma$ is reduced to
one-dimensional Gaussian measures.

Let $V = \R h$ with $S(h,h) = s > 0$ and consider the convolution algebra
with repsect to the basis $\{h\}$:
In terms of the multiplicative group elements $e^{ixh}$ ($x \in \R$), 
the density operator is expressed by $\int_\R e^{-sx^2/2} e^{ixh}\, dx$ in the group algebra,
which is Fourier transformed
to the multiplication function $\sqrt{\frac{2\pi}{s}} e^{-\xi^2/2s}$
with its $r$-th power $\left( \frac{2\pi}{s} \right)^{r/2} e^{-r\xi^2/2s}$ 
transformed back to the convolution algebra element
\[
  (2\pi)^{(r-1)/2} r^{-1/2} s^{(1-r)/2} \int_\R e^{-sx^2/2r}\, e^{ixh}\, dx. 
\]
Remark here that $h$ is a unit vector with respect to $S + \overline{S}$ if and only if $s = 1/2$.

We now compare this with the above formula for $\tau(\rho_S^r)$, where 
the seed function
$\frac{(s+t)^{r-1}|s-t|}{|s^r - t^r|}$ has a limit
\[
\lim_{s,t \to 1/2} \frac{(s+t)^{r-1}|s-t|}{|s^r - t^r|} = \frac{2^{r-1}}{r}. 
\]
Thus the trace formula for $\rho_S^r$ remains valid for degenerate $\sigma$ as well, in so far as 
$S + \overline{S}$ is non-degenerate. 

\begin{Theorem} Let $(V,\sigma)$ be a finite-dimensional 
alternating vector space and $S \in \sP(V,\sigma)$ 
be generic in the sense that $S + \overline{S}$ is non-degenerate.
Then, with respect to the Euclidean measure associated to $S + \overline{S}$,
$\rho_S^r$ ($r>0$) belongs to $L(V,\sigma) = C_0(V) \cap L^1(V)$ and is given by
$\rho_S^r = \tau(\rho_S^r) \rho_{S^{(r)}}$ with
\[
\tau(\rho_S^r) = 
(2\pi)^{(r-1)\dim V/2} 
\det\left(
\frac{(S+\overline{S})^{r-1} |S-\overline{S}|}{|S^r - {\overline S}^r|}
\right)^{1/2}.  
\]
\end{Theorem}


\begin{thebibliography}{20}

\bibitem{Ar}
H.~Araki, On quasifree states of CAR and Bogoliubov automorphisms, 
\textit{Publ.~RIMS}, 6(1970/1971), 385--442. 


\bibitem{AY} 
H.~Araki and S.~Yamagami, 
On quasi-equivalence of quasifree states of 
the canonical commutation relations, 
\textit{Publ.~RIMS}, 18(1982), 703--758.





\bibitem{BR2} 
O.~Bratteli and D.W.~Robinson, 
\textit{Operator Algebras and Quantum Statistical Mechanics II}, 
Springer-Verlag, 1979.

\bibitem{Fu}
  J.I.~Fujii, Kubo-Ando theory for convex functional means,
  \textit{Math.~Japon.}, 7(2002), 299--311.

\bibitem{HK}
  F.~Hiai and H.~Kosaki, \textit{Means of Hilbert Space Operators}, LNM~1820, Springer-Verlag, 2003. 
  
\bibitem{KA}
  F.~Kubo and T.~Ando, Means of positive linear operators, \textit{Math.~Ann.}, 248(1980), 205--224.

\bibitem{NW}
M.~Nagisa and S.~Wada, Operator monotonicity of some functions, 
\textit{Linear Algebra Appl.}, 486(2015), 389-408.

\bibitem{PW1} 
W.~Pusz and S.L.~Woronowicz, 
Functional calculus for sesquilinear froms and 
the purification map, 
\textit{Rep.~Math.~Phys.}, 
8(1975), 159--170.

\bibitem{PW2}
  W.~Pusz and S.L.~Woronowicz, 
  Form convex functions and the WYDL and other inequalities,
  \textit{Lett.~Math.~Phys.}, 2(1978), 505--512.

\bibitem{vN}
  J.~von Neumann, Die Eindeutigkeit der Schr\"odingerschen Operatoren, \textit{Math.~Ann.},
  104(1931), 570--578.

\bibitem{gqfs}
S.~Yamagami, Geometry of quasi-free states of CCR-algebras, 
\textit{Int.~J.~Math.}, 21(2010), 875--913. 
\end{thebibliography}
\end{document}